\documentclass[runningheads]{llncs}
\usepackage[T1]{fontenc}
\usepackage{graphicx}
\usepackage{amsmath}
\usepackage{amssymb}
\usepackage{array}
\usepackage[shortlabels]{enumitem}
\usepackage{hyperref}
\usepackage[nameinlink]{cleveref}
\usepackage{color}

\def\sp{\mathsf{\#P}}
\def\gapp{\mathsf{GapP}}

\def\spanp{\mathsf{SpanP}}

\def\np{\mathsf{NP}}
\def\conp{\mathsf{coNP}}

\def\spp{\mathsf{SPP}}
\def\pp{\mathsf{PP}}

\def\fp{\mathsf{FP}}

\def\p{\mathsf{P}}
\def\npsv{\mathsf{NPSV_{t}}}

\def\u{\mathsf{U}}
\def\s{\mathsf{S}}

\def\optp{\mathsf{OptP}}
\def\maxp{\mathsf{MaxP}}
\def\minp{\mathsf{MinP}}
\def\midp{\mathsf{MidP}}
\def\medp{\mathsf{MedP}}
\def\medpb{\mathsf{\overline{MedP}}}
\def\ktho{\mathsf{KthOutput}}
\def\kthos{\mathsf{KthOutput_{seq}}}
\def\kthw{\mathsf{KthWitness}}
\def\maxmcp{\mathsf{MaxMCP}}

\newtheorem{Theorem}{Theorem}[section]
\newtheorem{Lemma}[Theorem]{Lemma}
\newtheorem{Proposition}[Theorem]{Proposition}
\newtheorem{Definition}[Theorem]{Definition}
\newtheorem{Corollary}[Theorem]{Corollary}

\begin{document}
	\title{On the Complexity of Computing Outputs of a Metric Turing Machine}

	\author{Yaroslav Ivanashev} %\orcidID{0009-0001-5059-1844}} %\and
		%Second Author\inst{2,3}\orcidID{1111-2222-3333-4444} \and
		%Third Author\inst{3}\orcidID{2222--3333-4444-5555}}
	%
	\authorrunning{Y. Ivanashev}
	
	\institute{HSE University, Faculty of Computer Science, Moscow, 101000, Russia
		\email{yivanashev@hse.ru}}
	%\institute{Princeton University, Princeton NJ 08544, USA \and
		%Springer Heidelberg, Tiergartenstr. 17, 69121 Heidelberg, Germany
	\maketitle              % typeset the header of the contribution
	\begin{abstract}
		
		The classes $\midp$, $\medp$, and $\medpb$ contain functions that compute the median solution for certain types of problems. In this paper, for these classes we introduce analogous classes of functions that compute the $k$-th solution, where $k$ is an order function that depends on the input. We prove that the classes $\midp$, $\medp$, and $\medpb$ are polynomial-time 1-Turing inter-reducible with the corresponding classes, where the order function is from  $\fp$ or $\fp^{\sp}$. For $\medp$ we also prove that it coincides with the corresponding classes, where the order function is from $\fp$ or $\sp$. 
		For several inclusions between function classes we give equivalent inclusions between language classes. In particular, we establish inclusion relations between $\maxp$ and median classes $\midp$, $\medp$, and $\medpb$. We also prove that  $\npsv \subseteq \maxp \subseteq \fp^{\np}$ and both inclusions are proper if and only if $\np \neq \conp$.

		
	\end{abstract}

	\section{Introduction}
	
	Krentel introduced the class $\optp$ \cite{kre88} that consists of functions that are equal to the maximum or minimum output of a nondeterministic polynomial-time bounded Turing Machine (NPTM) that outputs an integer on every computation path (a NPTM with this property is called metric). He also proved that $\fp^{\optp[1]} = \fp^{\np}$ and studied complete problems for $\fp^{\np}$ under polynomial-time 1-Turing reductions. Toda defined a class $\midp$ that consists of functions that are equal to the median output of a metric NPTM \cite{tod94}, proved that $\fp^{\midp[1]} = \fp^{\sp}$, and studied
	complete problems for $\fp^{\sp}$ under polynomial-time 1-Turing reductions.
	Vollmer and Wagner defined two related median classes: $\medp$ is defined as $\midp$ except that the multiplicities of outputs are taken into account, and $\medpb$ consists of functions that are equal to the median witness of an NP problem \cite{vol93}. They also showed that these classes are polynomial-time 1-Turing inter-reducible: $\fp^{\midp[1]} = \fp^{\medp[1]} = \fp^{\medpb[1]} = \fp^{\sp}$. 
	In \Cref{sec-fpsp}, we give a simplified proof of this result. 
	Vollmer in \cite{vol94} studied complete problems under parsimonious reductions for the classes $\midp$, $\medp$, and $\medpb$.
	
	In this paper, for the median classes $\midp$, $\medp$, and $\medpb$ we define their analogs $\ktho(C)$, $\kthos(C)$, and $\kthw(C)$, which consist of functions that are equal to the $k$-th output or witness (instead of median output or witness) for a function $k \in C$. Toda used a class $\mathsf{KthP}$ \cite{tod94}, which is similar to $\ktho(\fp)$ but contains partial functions that take two arguments. We use classes of total functions of one argument in order to compare them with $\midp$, $\medp$, and $\medpb$. In fact, $\mathsf{KthP}$ would be a subset of $\ktho(\fp)$ if it contained only total functions, so $\ktho(\fp)$ can be considered as a generalization of $\mathsf{KthP}$ that admits different order functions $k \in \fp$. For example, the problem KthSAT from \cite{tod94} belongs to the classes $\ktho(\fp)$ and $\kthw(\fp)$ (the assignments of formulas should be encoded as integers).
	In \Cref{sec-mid}, for the classes $\medpb$, $\midp$, and $\medp$ we show that they are polynomial-time 1-Turing inter-reducible with their corresponding classes, where the function $k$ is either in $\fp$ or $\fp^{\sp}$: $\fp^{\sp} = \fp^{\medpb[1]} = \fp^{\kthw(\fp)[1]} = \fp^{\kthw(\fp^{\sp})[1]} = \fp^{\midp[1]} = \fp^{\ktho(\fp)[1]} = \fp^{\ktho(\fp^{\sp})[1]} = \fp^{\medp[1]} = \fp^{\kthos(\fp)[1]}\\ = \fp^{\kthos(\fp^{\sp})[1]}$. In \Cref{sec-med}, for the class $\medp$ we prove that it coincides with the corresponding classes $\kthos(\fp)$ and $\kthos(\sp)$.
	
	In \Cref{sec-inc}, for several inclusions between function classes we give equivalent inclusion between language classes. In \cite{vol93}, Vollmer and Wagner proved that $\medpb \subseteq \midp \cap \medp$ and the other inclusions between these classes have implausible implications. For some of these implications we show that they hold in the opposite direction. We prove that $\fp^{\np \cap \conp} \subseteq \maxp \subseteq \fp^{\np}$ and both inclusions are proper if and only if $\np \neq \conp$. We also prove that $\maxp$ contains the classes $\midp$, $\medp$, and $\medpb$ if and only if $\pp = \np$, and $\medpb$ contains the classes $\spanp$ and $\maxp$ if and only if $\np \subseteq \spp$.
		
	\section{Preliminaries}
	
	In this paper, we use a standard definition of a nondeterministic polynomial-time bounded Turing machine (NPTM) and binary alphabet $\Sigma = \{0, 1\}$. We assume that the reader is familiar with basic complexity theory notions and classes $\np$, $\pp$, $\spp$ \cite{fen94},  $\sp$ \cite{val79}, $\gapp$ \cite{fen94}, $\spanp$ \cite{kob89}. 
	
	\begin{Definition} [\cite{kre88}]
		A NPTM is called metric if on every branch it outputs a binary number from $\mathbb{Z}$ and accepts.
	\end{Definition}
	
	In the definition above, the outputs of a NPTM can be negative numbers (in some reasonable encoding). 
	
	Krentel defined the class $\optp$ \cite{kre88} as a union of the classes $\minp$ and $\maxp$ that we define below. 
	
	\begin{Definition}
		\textnormal{\cite{kre88}} $\maxp = \{max_{M} \ | \ M$ is a metric NPTM$\}$, where $max_{M}(x)$ is the maximum value in a set of all outputs of $M$ on input $x \in \Sigma^{*}$.	More precisely, if $\{w_{1}, \dots, w_{k}\}$ is the set of all outputs of $M$ on input $x$ and if $w_{1} < \ldots < w_{k}$, then $max_{M}(x) = w_{k}$.
	\end{Definition}
	
	The class $\minp$ is defined analogously. In this paper, we will use only the class $\maxp$. All results that we prove for this class also hold for the class $\minp$.
		
	Below we give definitions of median classes that we will consider in this paper.
	
	\begin{Definition}
		\begin{enumerate}
			\item  \textnormal{\cite{tod94}} $\midp = \{mid_{M} \ | \ M$ is a metric NPTM$\}$, where $mid_{M}(x)$ is the median value in a set of all outputs of $M$ on input $x \in \Sigma^{*}$.	More precisely, if $\{w_{1}, \dots, w_{k}\}$ is the set of all outputs of $M$ on input $x$ and if $w_{1} < \ldots < w_{k}$, then $mid_{M}(x) = w_{\lfloor(k+1)/2\rfloor}$.
			\item  \textnormal{\cite{vol93}} $\medp = \{med_{M} \ | \ M$ is a metric NPTM$\}$, where $med_{M}(x)$ is the median value in a sequence of all outputs of $M$ on input $x \in \Sigma^{*}$. More precisely, if $w_{1} \leqslant \ldots \leqslant w_{k}$ is the sequence of all outputs (with multiplicities) of $M$ on input $x$, then $med_{M}(x) = w_{\lfloor(k+1)/2\rfloor}$.		
			\item  \textnormal{\cite{vol93}} The class $\medpb$ consists of all functions $f$ for which there exist a language $A \in \p$ and functions $l, r \in \fp$ ($\Sigma^{*} \rightarrow \mathbb{Z}$) such that for every $x \in \Sigma^{*}$:
			\begin{align*}
				f(x) = \begin{cases}
					y_{\lfloor(k+1)/2\rfloor}, &\text{if the set $W$ is not empty,}\\
					l(x) - 1, &\text{$otherwise$,}
				\end{cases}
			\end{align*}
			where $W = \{y \in \mathbb{Z} \ | \ l(x) \leqslant y \leqslant r(x) \wedge (x, y) \in A\} = \{y_{1}, \dots, y_{k}\}$ and $y_{1} < \ldots < y_{k}$.
		\end{enumerate}
	\end{Definition}
	
	In \cite[Proposition 2.6]{vol93}, Vollmer and Wagner proved that $\medpb$ is a subset of $\midp$ and  $\medp$. By definition in \cite{vol93}, the class $\medpb$ contains partial functions, but the proof of \cite[Proposition 2.6]{vol93} can be used for \Cref{med-unc}, where we consider $\medpb$ as a class of total functions.
	\Cref{med-inc} shows that the other inclusions between these classes are not likely to hold.
	
	\begin{Proposition} [\cite{vol93}] \label[Proposition]{med-unc}
		$\medpb \subseteq \midp \cap \medp$.
	\end{Proposition} 
	
	For the median classes above we define their analogs, where the functions compute outputs or witnesses of some specific order.
	
	\begin{Definition}
		Let $C$ be a class of functions $\Sigma^{*} \rightarrow \mathbb{N}$.
		\begin{enumerate}
			\item The class $\ktho(C)$ consists of all functions $f$ for which there exist a metric NPTM $M$ and a function $g \in C$ such that for every $x \in \Sigma^{*}$:
			\begin{align*}
				f(x) = \begin{cases}
					w_{g(x)}, &\text{if $M$ has at least $g(x)$ outputs on input $x$,}\\
					2^{p(|x|)}, &\text{$otherwise$,}
				\end{cases}
			\end{align*}
			where $\{w_{1}, \dots, w_{k}\}$ is the set of all outputs of $M$ on input $x$ and $w_{1} < \ldots < w_{k}$, and $p$ is a polynomial that on every input is greater than the running time of $M$.
			\item The class $\kthos(C)$ consists of all functions $f$ for which there exist a metric NPTM $M$ and a function $g \in C$ such that for every $x \in \Sigma^{*}$:
			\begin{align*}
				f(x) = \begin{cases}
					w_{g(x)}, &\text{if $M$ has at least $g(x)$ paths on input $x$,}\\
					2^{p(|x|)}, &\text{$otherwise$,}
				\end{cases}
			\end{align*}
			where $w_{1} \leqslant \ldots \leqslant w_{k}$ is the sequence of all outputs (with multiplicities) of $M$ on input $x$, and $p$ is a polynomial that on every input is greater than the running time of $M$.
			\item The class $\kthw(C)$ consists of all functions $f$ for which there exist a language $A \in \p$, functions $l, r \in \fp$ ($\Sigma^{*} \rightarrow \mathbb{Z}$), and a function $g \in C$ such that for every $x \in \Sigma^{*}$:
			\begin{align*}
				f(x) = \begin{cases}
					y_{g(x)}, &\text{if $|W| \geqslant g(x)$,}\\
					l(x) - 1, &\text{$otherwise$,}
				\end{cases}
			\end{align*}
			where $W = \{y \in \mathbb{Z} \ | \ l(x) \leqslant y \leqslant r(x) \wedge (x, y) \in A\} = \{y_{1}, \dots, y_{k}\}$ and $y_{1} < \ldots < y_{k}$. 
		\end{enumerate}
	\end{Definition}
	
	In the definitions of $\medpb$, $\ktho$, $\kthos$, and $\kthw$, the values $l(x) - 1$ and $2^{p(|x|)}$ are used to indicate a number that cannot be an output or witness of the considered machine. For our results it is only important that these values are polynomial-time computable, and some proofs use that these values are always greater or less than possible outputs or witnesses.
	
	\begin{Definition} [\cite{boo84,boo85}]
		The class $\npsv$ consists of all functions $f$ such that there exists a NPTM that on every input $x \in \Sigma^{*}$ has at least one accepting path and outputs the value of $f(x)$ on every accepting path.
	\end{Definition}
	
	The proof of \Cref{npsv-fp} is given in \cite[Proposition 2.4]{boo85}. 
	
	\begin{Proposition} \label[Proposition]{npsv-fp}
		\textnormal{\cite{sel94}} $\npsv = \fp^{\np \cap \conp}$.
	\end{Proposition}
	
	In this paper, we assume that the class $\npsv$ contains only integer-valued functions in order to compare it with other classes of integer-valued functions.
	
	\section{Characterization of $\fp^{\sp}$} \label{sec-fpsp}
	
	In \cite{tod94}, Toda gave a characterization of the class $\fp^{\sp}$ using $\midp$ functions: $\fp^{\sp} = \fp^{\midp[1]}$. Vollmer and Wagner used the same technique to show that the class $\medpb$ can be used in this characterization \cite{vol93}. We give a simplified version of this proof below.
	
	\begin{Lemma} \label[Lemma]{pp}
		For every language $L \in \pp$ there exist a NPTM $M$ and a function $f\in \fp$ such that for every $x \in \Sigma^{*}$:
		\begin{enumerate}
			\item $x \in L \Leftrightarrow$ $M$ has more accepting paths than rejecting paths on input $x$.
			\item The number of paths of $M$ on input $x$ is odd and depends only on the length of $x$.
			\item $f(x)$ equals the number of paths of $M$ on input $x$.
		\end{enumerate}
	\end{Lemma}
	
	\begin{proof}
		Let $L \in \pp$. By definition of $\pp$, there exists a NPTM that satisfies the first condition, and it is well known that this machine can be modified so that on every input all paths have the same number of nondeterministic steps. After that one rejecting path should be added to make the number of paths odd. \qed
	\end{proof}

	\begin{Theorem} [\cite{vol93}] \label{fpsp}
		$\fp^{\sp} = \fp^{\medpb[1]}$.
	\end{Theorem}
	
	\begin{proof}
		\begin{enumerate}
			\item $\medpb \subseteq \fp^{\sp}$: the middle witness can be found by binary search on the set of witnesses. 
			\item $\fp^{\sp} \subseteq \fp^{\medpb[1]}$: It suffices to prove $\fp^{\pp} \subseteq \fp^{\medpb[1]}$, because $\fp^{\sp} = \fp^{\pp}$. Let $M$ be a polynomial-time bounded Turing machine with an oracle $L \in \pp$, and $N$ be a NPTM that witnesses $L$ and satisfies the conditions of \Cref{pp}. The machine $M$ can be modified so that on every input of length $n$ it has only  queries of length $q(n)$ and has exactly $p(n)$ queries, where $p$ and $q$ are polynomials. In this proof, we will use binary strings of fixed polynomial length to represent integer witnesses from the definition of $\medpb$. This set of witnesses for every $x \in \Sigma^{*}$ contains strings of the form $a_{1}q_{1}w_{1} \ldots a_{p(|x|)}q_{p(|x|)}w_{p(|x|)}$, where for every $i$, $a_{i}$ is an oracle answer to the $i$-th query $q_{i}$ of $M$ on input $x$ and $w_{i}$ is a path of $N$ on input $q_{i}$. For every $i$ the $i$-th query $q_{i}$ is determined by the previous answers $a_{1}, \ldots, a_{i-1}$, which can be different from the correct answers. All paths $w_{i}$ should be encoded so that they have the same polynomial length.  Moreover, this set contains only strings of this form, where for every $i$, $a_{i} = 1$ if and only if $w_{i}$ is an accepting path. The middle witness from this set contains the correct answer $a_{1}$, because the machine $N$ has more accepting paths than rejecting paths on the input $q_{1}$ if and only if the correct answer is 1. The second answer $a_{2}$ in the middle witness is also correct, because the second query $q_{2}$ is determined correctly in the middle witness, and every prefix $a_{1}q_{1}w_{1}$ can be extended to the same number of prefixes $a_{1}q_{1}w_{1}a_{2}q_{2}w_{2}$. By the same argument, the middle witness contains all correct oracle answers.
			\qed
		\end{enumerate}
	\end{proof}

	In \cite{tod94}, Toda proved that $\midp \subseteq \fp^{\sp}$, which follows from the main theorem of \cite{tod92}. Vollmer and Wagner showed that $\medp \subseteq \fp^{\sp}$ in \cite[Corollary 5.2]{vol93}. This gives the following corollary:
	
	\begin{Corollary} [\cite{tod94,vol93}] \label[Corollary]{1-t}
		$\fp^{\sp} = \fp^{\medpb[1]} = \fp^{\midp[1]} = \fp^{\medp[1]}$.
	\end{Corollary}
	
	\section{Computing the K$^{\textnormal{\textbf{th}}}$ Witness and the K$^{\textnormal{\textbf{th}}}$ Value in a Set of Outputs} \label{sec-mid}
	
	In this section, we prove that the classes $\medpb$, $\midp$, and $\medp$ are polynomial-time 1-Turing inter-reducible with the corresponding classes, where the order function is either in $\fp$ or $\fp^{\sp}$. 
	
	\begin{Lemma} \label[Lemma]{medpb}
		$\fp^{\medpb[1]} = \fp^{\kthw(\fp)[1]} = \fp^{\kthw(\fp^{\sp})[1]}$.
	\end{Lemma}
	
	\begin{proof}
		It suffices to prove that $\kthw(\fp^{\sp}) \subseteq \fp^{\sp}$ and $\fp^{\sp} \subseteq$\\ $\fp^{\kthw(\fp)[1]}$, because $\fp^{\sp} = \fp^{\medpb[1]}$ by \Cref{fpsp}.
		\begin{enumerate}
			\item $\kthw(\fp^{\sp}) \subseteq \fp^{\sp}$: as in the first part of \Cref{fpsp}, the $g(x)$-th witness for a function $g \in \fp^{\sp}$ can be found by binary search. To compute the value of $g(x)$ another $\sp$ oracle is needed, but one $\sp$ oracle can be used to compute two $\sp$ functions \cite{iva26}.
			\item $\fp^{\sp} \subseteq \fp^{\kthw(\fp)[1]}$: the proof follows from the second part of \Cref{fpsp}, because the number of witnesses in this proof is polynomial-time computable.
			\qed
		\end{enumerate}
	\end{proof}
	
	The proof of $\fp^{\midp[1]} = \fp^{\ktho(\fp)[1]}$ is given in \cite{tod94} in a different notation. 
	
	\begin{Lemma} \label[Lemma]{midp}
		$\fp^{\midp[1]} = \fp^{\ktho(\fp)[1]} = \fp^{\ktho(\fp^{\sp})[1]}$.
	\end{Lemma}
	
	\begin{proof}
		It suffices to prove that $\ktho(\fp^{\sp}) \subseteq \fp^{\sp}$ and $\fp^{\sp} \subseteq$\\ $\fp^{\ktho(\fp)[1]}$, because $\fp^{\sp} = \fp^{\midp[1]}$ by \Cref{1-t}.
		\begin{enumerate}
			\item $\ktho(\fp^{\sp}) \subseteq \fp^{\sp}$: the proof is analogous to the first part of \Cref{medpb}, except that the binary search on the set of outputs requires queries to $\spanp$. It follows from the main theorem of \cite{tod92} that $\spanp \subseteq \fp^{\sp}$.
			\item $\fp^{\sp} \subseteq \fp^{\ktho(\fp)[1]}$: the proof follows from the second part of \Cref{fpsp}, because there exists a NPTM that outputs all witnesses from this proof on distinct paths without generating other paths, and the number of witnesses is polynomial-time computable.
			\qed
		\end{enumerate}
	\end{proof}
	
	For the class $\medp$ analogous lemma can be proved similarly. The next theorem summarizes results of the last two sections. 
	
	\begin{Theorem} \label{all}
		$\fp^{\sp} = \fp^{\medpb[1]} = \fp^{\kthw(\fp)[1]} = \fp^{\kthw(\fp^{\sp})[1]} = \fp^{\midp[1]} = \fp^{\ktho(\fp)[1]} = \fp^{\ktho(\fp^{\sp})[1]} = \fp^{\medp[1]} = \fp^{\kthos(\fp)[1]} = \fp^{\kthos(\fp^{\sp})[1]}$
	\end{Theorem}
	
	\section{Computing the K$^{\textnormal{\textbf{th}}}$ Value in a Sequence of Outputs} \label{sec-med}
	
	For the class $\medp$, we give a stronger result than \Cref{all}. In the next theorem, we prove that the classes $\medp$, $\kthos(\fp)$, and $\kthos(\sp)$ are equal. 	
	
	\begin{Definition}
		The class $\maxmcp$ consists of all functions $f$ for which there exist a language $A \in \pp$ and functions $l, r \in \fp$ ($\Sigma^{*} \rightarrow \mathbb{Z}$) such that:
		\begin{enumerate}
		\item For every $x \in \Sigma^{*}$:
		\begin{align*}
			f(x) = \begin{cases}
				y_{k}, &\text{if the set $W$ is not empty,}\\
				l(x) - 1, &\text{$otherwise$,}
			\end{cases}
		\end{align*}
		where $W = \{y \in \mathbb{Z} \ | \ l(x) \leqslant y \leqslant r(x) \wedge (x, y) \in A\} = \{y_{1}, \dots, y_{k}\}$ and $y_{1} < \ldots < y_{k}$. 
		\item For every $x \in \Sigma^{*}$, if $(x, y) \in A$ and $l(x) \leqslant y \leqslant r(x)$, then $(x, y - 1) \in A$.
		\end{enumerate}
	\end{Definition}
	
	A proof of $\medp = \maxmcp$ is sketched in an unpublished manuscript \cite[Lemma 4.10]{vol96}.
	
	\begin{Theorem} \label{medp}
		$\medp = \maxmcp = \kthos(\fp) = \kthos(\sp)$.
	\end{Theorem}
	
	\begin{proof}
		\begin{enumerate}
			\item $\kthos(\sp) \subseteq \medp$: Let $f \in \kthos(\sp)$ and $g \in \sp$ from the definition of $\kthos(\sp)$. To construct the required function from $\medp$, for every $x \in \Sigma^{*}$ the sequence of outputs $w_{1} \leqslant \ldots \leqslant w_{k}$ will be extended so that the $g(x)$-th output in the initial sequence is the middle output in the extended sequence. To the sequence $w_{1} \leqslant \ldots \leqslant w_{k}$ $k$ smaller numbers should be added to the left and the number $2^{p(|x|)}$ with multiplicity $2g(x)$ should be added to the right, where $p$ is a fixed polynomial from the definition of $\kthos$.
			\item $\medp \subseteq \maxmcp$: Let $h \in \medp$, then $h(x) = med_{M}(x)$ for a metric NPTM $M$. There exist functions $l, r \in \fp$ such that for every $x \in \Sigma^{*}$ all outputs of $M$ on input $x$ are in the interval $[l(x), r(x)]$. For every $x \in \Sigma^{*}$, let $S(x)$ be a set $\{(s, t) \ | \ s$ is a path of $M$ on input $x$ that outputs $t\}$. Let $f$ and $g$ be functions such that for every $x \in \Sigma^{*}$ and $y \in \mathbb{Z}$ such that $l(x) \leqslant y \leqslant r(x)$, $f(x, y)$ equals the number of elements $(s, t) \in S(x)$ such that $t < y$, and $g(x, y)$ equals the number of elements $(s, t) \in S(x)$ such that $t \geqslant y$. The functions $f$ and $g$ are in $\sp$, thus the language $A = \{(x, y) \ | \ f(x, y) < g(x, y)\}$ is in $\pp$. The median value in the sequence of outputs of $M$ on input $x$ equals the maximal element in the set $\{y \in \mathbb{Z} \ | \ l(x) \leqslant y \leqslant r(x) \wedge (x, y) \in A\}$, and this set has the second property from the definition of $\maxmcp$.		
			\item $\maxmcp \subseteq \medp$: Let $f \in \maxmcp$, and $A \in \pp$, $l, r \in \fp$ be defined as in the definition of $\maxmcp$. Let $N$ be a NPTM that witnesses $A$ and satisfies the conditions of \Cref{pp}. $N$ also should be modified so that for every $x \in \Sigma^{*}$ and $y \in \mathbb{Z}$ such that $l(x) \leqslant y \leqslant r(x)$ it has $2^{p(|x|)} + 1$ paths on input $(x, y)$ for a polynomial $p$. To define the required function from $\medp$, we will describe a NPTM $M$ that on every input $x \in \Sigma^{*}$ makes a binary search on the set $\{y \in \mathbb{Z} \ | \ l(x) \leqslant y \leqslant r(x)\}$.
			On input $x \in \Sigma^{*}$, $M$ first simulates the machine $N$ on input $(x, \lfloor (l(x)+r(x))/2 \rfloor)$. If $N$ accepts, $M$ simulates $N$ on input 
			$(x, \lfloor (l(x)+3r(x))/4 \rfloor)$. If $N$ rejects, $M$ simulates $N$ on input $(x, \lfloor (3l(x)+r(x))/4 \rfloor)$. $M$ continues binary search until it has the same number of simulations on every path, and outputs $y$ if the last simulation was on input $(x, y)$. If $N$ rejects the input $(x, l(x))$, $M$ outputs $l(x) - 1$. At every simulation of $N$, more than a half of paths of $N$ terminate with a correct decision of $(x, y) \in A$, and accepting paths always lead to larger outputs than rejecting paths. Therefore, the path with the middle output has correct decisions of $(x, y) \in A$, and outputs the value of $f(x)$.
			\item $\maxmcp \subseteq \kthos(\fp)$: follows from the proof of $\maxmcp \subseteq \medp$, because the number of paths of $M$ is polynomial-time computable. \qed
		\end{enumerate}
	\end{proof}
	
	\section{Inclusions Between Function Classes} \label{sec-inc}
	
	In this section, we give characterizations for several inclusions between function classes in terms of inclusions between language classes. 
	
	Below we give definitions of operators $\u$ and $\s$ that map classes of functions to classes of languages.
	
	\begin{Definition} [\cite{vol93}]
		Let $F$ be a class of functions $\Sigma^{*} \rightarrow \mathbb{Z}$. 
		\begin{enumerate}
			\item The class $\u \cdot F$ consists of all languages $L$ such that there exists a function $f \in F$ such that for every $x \in \Sigma^{*}$:
			\begin{align*}
				x \in L \Leftrightarrow f(x) = 1, \\
				x \notin L \Leftrightarrow f(x) = 0.
			\end{align*}
			\item The class $\s \cdot F$ consists of all languages $L$ such that there exist functions $f \in F$ and $g \in \fp$  such that for every $x \in \Sigma^{*}$:
			\begin{align*}
				x \in L \Leftrightarrow f(x) > g(x).
			\end{align*}
		\end{enumerate}
	\end{Definition}
	
	\begin{Proposition} [\cite{vol93}] \label[Proposition]{oper}
		\begin{enumerate} 
			\item $\u \cdot \medpb = \spp$, $\s \cdot \medpb = \pp$.
			\item $\u \cdot \medp = \pp$, $\s \cdot \medp = \pp$.
			\item $\u \cdot \midp = \spp^{\np}$, $\s \cdot \midp = \pp^{\np}$.
		\end{enumerate}
	\end{Proposition}
	
	In \cite{vol93}, it is proved that $\midp$ and $\medp$ contain $\medpb$. \Cref{med-inc} gives characterizations for other inclusions between these classes. We leave a characterization of $\midp = \medpb$ as an open question. The implications from left to right were proved in \cite{vol93}. For the statements on the right there exist separating oracles that also give separations for the statements on the left \cite{vol93}. 
	
	\begin{Theorem} \label{med-inc}
		\begin{enumerate}
			\item If $\midp = \medpb$, then $\np \subseteq \spp$.
			\item $\medp = \medpb$ if and only if $\pp = \spp$.
			\item $\medp \subseteq \midp$ if and only if $\pp \subseteq \spp^{\np}$. 
			\item $\midp \subseteq \medp$ if and only if $\pp^{\np} = \pp$.
		\end{enumerate}
	\end{Theorem}
	
	\begin{proof}
		The implications from left to right follow from \Cref{oper} \cite{vol93}. $\spp^{\np} = \spp$ is equivalent to $\np \subseteq \spp$, because $\spp^{\spp} = \spp$ \cite{fen94}. 
		
		In this proof we also use $\gapp^{\spp} = \gapp$ \cite{fen94}, which can be relativized to any oracle \cite{mah94}. In particular, $\gapp^{\spp^{\np}} = \gapp^{\np}$. The proofs of $\gapp \subseteq \medpb$ and $\sp^{\np} - \fp \subseteq \midp$ are given in \cite{vol93}.
		
		If $\pp = \spp$, then $\medp \subseteq \fp^{\pp} \subseteq \gapp^{\pp} \subseteq \gapp^{\spp} \subseteq \gapp \subseteq \medpb$.
		
		If $\pp \subseteq \spp^{\np}$, then $\medp \subseteq \fp^{\pp} \subseteq \gapp^{\pp} \subseteq \gapp^{\spp^{\np}} \subseteq \gapp^{\np} \subseteq \sp^{\np} - \fp \subseteq \midp$.
		
		If $\pp^{\np} = \pp$, then $\medp^{\np} \subseteq \maxmcp$, because the proof of $\medp \subseteq \maxmcp$ relativizes. In \cite{vol93}, it is shown that $\midp \subseteq \medp^{\np}$. 
		\qed
	\end{proof}
	
	In \cite{kre88}, Krentel proved that $\fp^{\maxp[1]} = \fp^{\np}$, which implies that $\maxp \subseteq \fp^{\np}$. We show in the next proposition that $\maxp$ can be placed between two similar classes $\fp^{\np \cap \conp}$ and $\fp^{\np}$ and both inclusions are proper unless $\np = \conp$. By \Cref{npsv-fp}, $\fp^{\np \cap \conp}$ is equal to the class $\npsv$. In the next propositions, $\minp$ can be used instead of $\maxp$.
	
	\begin{Proposition}
		\begin{enumerate}
			\item $\npsv \subseteq \maxp \subseteq \fp^{\np}$.
			\item $\npsv = \maxp$ if and only if $\np = \conp$.
			\item $\maxp = \fp^{\np}$ if and only if $\np = \conp$.
		\end{enumerate}
	\end{Proposition}
	
	\begin{proof}
		$\npsv \subseteq \maxp$, because the machine can output sufficiently small values instead of rejecting. If $\np = \conp$, then  $\fp^{\np \cap \conp} = \maxp = \fp^{\np}$. If either $\npsv = \maxp$ or $\maxp = \fp^{\np}$, then for every function $f \in \maxp$ the function $1 - f$ also belongs to $\maxp$, which implies that the class $\np = \u \cdot \maxp$ is closed under complement. 
		\qed
	\end{proof}
	
	In the next propositions, we establish relations between $\maxp$ and classes $\midp$, $\medp$, and $\medpb$.
	
	\begin{Proposition}
		$\maxp \subseteq \medp \cap \midp$.
	\end{Proposition}
	
	\begin{proof}
		Let $M$ be a NPTM that witnesses $f \in \maxp$ for a function $f$. The required NPTM simulates $M$ and instead of outputting the value $n$ outputs the values $n$ and $n + m$ for a sufficiently large $m$ that depends only on the input.
		\qed
	\end{proof}

	\begin{Proposition}
		The following statements are equivalent:
		\begin{enumerate}
			\item $\medp = \maxp$.
			\item $\midp = \maxp$.
			\item $\medpb \subseteq \maxp$. 
			\item $\np = \pp$.
		\end{enumerate}
	\end{Proposition}
	
	\begin{proof}
		The first three statements imply $\pp = \np$, because $\np = \s \cdot \maxp$ and $\pp = \s \cdot \medpb = \s \cdot \medp \subseteq \s \cdot \midp$. If $\np = \pp$, then the classes $\medp$, $\midp$, and $\medpb$ are subsets of the class $\fp^{\pp} = \maxp = \fp^{\np \cap \conp}$.
		\qed
	\end{proof}
	
	The class $\medpb$ contains the class $\gapp$ \cite{vol93}, while the class $\midp$ also contains the classes $\spanp$ \cite{vol93} and $\maxp$. In the next proposition, we show that $\medpb$ is not likely to contain the classes $\spanp$ and $\maxp$. The equivalence between $\spanp \subseteq \gapp$ and $\np \subseteq \spp$ can be shown similarly \cite{iva26}. 
	
	\begin{Proposition} \label[Proposition]{np-spp}
		The following statements are equivalent:
		\begin{enumerate}
			\item $\spanp \subseteq \medpb$.
			\item $\maxp \subseteq \medpb$.
			\item $\np \subseteq \spp$.
		\end{enumerate}
	\end{Proposition}
	
	\begin{proof}
		The first two statements imply $\np \subseteq \spp$, because $\np = \u \cdot \spanp = \u \cdot \maxp$ and $\spp = \u \cdot \medpb$. If $\np \subseteq \spp$, then $\spanp \cup \maxp \subseteq \sp^{\np} \subseteq \gapp^{\spp} \subseteq \gapp \subseteq \medpb$. \qed
	\end{proof}

	%\subsubsection{\ackname}
	
	%This work is an output of a research project (HSE-BR-2025-024) implemented as part of the Basic Research Program at HSE University.
	
	\bibliographystyle{splncs04}
	\bibliography{outputs-bibliography}

\end{document}